\documentclass[pra,aps,superscriptaddress,footinbib,twocolumn]{revtex4-2}

\usepackage{amsmath}        

\usepackage[colorlinks=true,citecolor=blue,linkcolor=magenta]{hyperref}
\usepackage[braket]{qcircuit}
\usepackage{appendix}
\usepackage{amsfonts}
\usepackage{amssymb}
\usepackage{amscd}
\usepackage{enumerate}
\usepackage{epsfig}
\usepackage{algpseudocode}
\usepackage{bm}
\usepackage{xcolor}
\usepackage{amsthm}
\usepackage{framed}
\usepackage{multirow}
\usepackage{mathrsfs,amssymb}
\usepackage{latexsym} 
\usepackage{physics}
\usepackage{epstopdf}
\usepackage[skins]{tcolorbox}
\newtcolorbox[auto counter]{tbox}[2][]{%
    enhanced, float=hbt, drop fuzzy shadow southeast,
    colback=white!5!white, colframe=white!50!black,
    width= .97\columnwidth,sharp corners, boxrule=0.8pt,
    title={Table \thetcbcounter: #2}, #1
}
\usepackage[linesnumbered,boxed,ruled,vlined,rightnl]{algorithm2e}
\usepackage{bbm}
\usepackage{graphicx}   
\usepackage{subcaption} 
\usepackage{float}      
\newtheorem{lemma}{Lemma}
\newtheorem{theorem}{Theorem}[section]

\begin{document}
\title{Practical Homodyne Shadow Estimation}

\author{Ruyu Yang}
\affiliation{State Key Lab of Processors, Institute of Computing Technology,
Chinese Academy of Sciences, 100190, Beijing, China}
\affiliation{Graduate School of China Academy of Engineering Physics, Beijing 100193, China}

\author{Xiaoming Sun}
\email{
sunxiaoming@ict.ac.cn
}
\affiliation{State Key Lab of Processors, Institute of Computing Technology,
Chinese Academy of Sciences, 100190, Beijing, China}
\affiliation{School of Computer Science and Technology, University of Chinese Academy of Sciences, Beijing, 100049,  China.}

\author{Hongyi Zhou}
\email{
zhouhongyi@ict.ac.cn
}

\affiliation{State Key Lab of Processors, Institute of Computing Technology,
Chinese Academy of Sciences, 100190, Beijing, China}

\begin{abstract}
Shadow estimation provides an efficient framework for estimating observable expectation values using randomized measurements. While originally developed for discrete-variable systems, its recent extensions to continuous-variable (CV) quantum systems face practical limitations due to idealized assumptions of continuous phase modulation and infinite measurement resolution. In this work, we develop a practical shadow estimation protocol for CV systems using discretized homodyne detection with a finite number of phase settings and quadrature bins. We construct an unbiased estimator for the quantum state and establish both sufficient conditions and necessary conditions for informational completeness within a truncated Fock space up to $n_{\mathrm{max}}$ photons. We further provide a comprehensive variance analysis, showing that the shadow norm scales as $\mathcal{O}(n_{\mathrm{max}}^4)$, improving upon previous $\mathcal{O}(n_{\mathrm{max}}^{13/3})$ bounds. Our work bridges the gap between theoretical shadow estimation and experimental implementations, enabling robust and scalable quantum state characterization in realistic CV systems.
\end{abstract}

\maketitle

\section{Introduction}

Shadow estimation has emerged as a powerful framework for efficiently estimating expectation values of quantum observables \cite{huang2020predicting, elben2023randomized, roth2018recovering, flammia2012quantum, paini2023introduction}. It constructs succinct classical descriptions ("classical shadows") of quantum states using randomized measurements across various bases. This approach enables the simultaneous estimation of expectation values for a large number of observables with provably optimal logarithmic sample complexity scaling, providing rigorous bounds on estimator variance. Its efficiency in sample complexity makes it widely applied to various tasks, including entanglement detection, error mitigation, and verifying complex quantum devices \cite{elben2018probing, vermersch2018probing, huggins2021virtual, endo2021hybrid, huang2022provably, brydges2019probing}.

Recent work has extended the classical shadow formalism to continuous-variable (CV) quantum optical systems \cite{becker2024classical, gandhari2024precision, aarts2023classical, qu2023classical, braunstein2005quantum, weedbrook2012gaussian}. Unlike finite-dimensional systems, CV states reside in an infinite-dimensional Fock space, making direct application of shadow estimation challenging \cite{walls2008quantum, gerry2005introductory}. Existing approaches address this by imposing a photon number cutoff, $n_{\mathrm{max}}$, confining the state to a finite-dimensional subspace. Within this subspace, protocols based on randomized homodyne detection have been developed \cite{gandhari2024precision, d1994detection, leonhardt1997measuring, lvovsky2009continuous}, and sample complexity bounds for achieving target estimation accuracy have been established.

However, practical implementations of these CV shadow protocols face significant experimental limitations. They crucially rely on idealized assumptions about the homodyne detection process: (1) the local oscillator (LO) phase $\theta$ is modulated continuously and sampled uniformly from $[0, 2\pi)$, and (2) the quadrature measurement outcome $x_{\theta}$ is resolved with infinite precision, yielding continuous outcomes. In reality, experimental constraints necessitate approximations: the LO phase is modulated discretely, typically set to $N$ values $\theta_k = 2k\pi/N$ ($k=0,\ldots,N-1$), and the quadrature axis is partitioned into $M$ finite-width bins $\{I_i\}$ due to finite detector resolution and analog-to-digital conversion. Consequently, the measurement is described by a discrete set of Positive Operator-Valued Measure (POVM) elements $\{\Pi_{i,k}\}$, corresponding to detecting an outcome in bin $I_i$ at phase setting $\theta_k$. Existing theoretical frameworks assuming continuous phase and quadrature measurements are not directly applicable to this experimentally relevant, discretized scenario \cite{cao2015discrete, smithey1993measurement, breitenbach1997measurement, grosshans2002continuous, madsen2012continuous}.


In this work, we bridge this gap by developing a shadow estimation protocol tailored to practical homodyne detection with discrete phase settings and binned quadrature outcomes. We construct an unbiased estimator based on the discrete POVM $\{\Pi_{i,k}\}$ and rigorously establish both sufficient (Theorem~\ref{theorem:1}: $N \geq 2n_{\mathrm{max}} + 1$, $M \geq n_{\mathrm{max}} + 1$) and necessary (Theorem~\ref{theorem:2}) conditions for the POVM to be informationally complete within the truncated $n_{\mathrm{max}}$-photon subspace. Furthermore, we perform a comprehensive variance analysis for estimating arbitrary observables (Theorem~\ref{theorem:3}), showing that the shadow norm $\|X\|_E$ scales as $\mathcal{O}(n_{\mathrm{max}}^4)$ under the informational completeness conditions, improving upon prior $\mathcal{O}(n_{\mathrm{max}}^{13/3})$ bounds derived from worst-case trace distance arguments. Numerical simulations illustrate how the variance varies with the number of discrete quadrature bins $M$, the number of discrete phases $N$, and the photon number cutoff $n_{\mathrm{max}}$ for a given quantum state and observable. This work provides a rigorous and experimentally feasible framework for estimating expectation values using classical shadows.

\section{POVM Representation of discrete Homodyne Detection}\label{sec:discretePOVM}

In practical implementations, homodyne detection involves discretizing both the quadrature measurement outcomes and the local oscillator (LO) phases to accommodate finite resolution and experimental constraints. Specifically, the LO phase $\theta$ is sampled at discrete values, typically chosen as
\begin{align}
\theta_k &= \frac{2k\pi}{N}
\end{align}
for $k = 0, 1, \ldots, N-1$, where $N$ is the total number of phase settings. Besides, the continuous quadrature variable $x$ is divided into $M$ bins $\{I_i\}$ such that each bin $I_i$ spans the range $[x_i, x_{i+1})$ for $i = 1, 2, \ldots, M$.
This discretization facilitates data collection and processing while approximating the continuous measurement scenario.

For each combination of quadrature bin $I_i$ and phase $\theta_k$, the corresponding measurement operator $\Pi_{i,k}$ is defined as
\begin{align}
\Pi_{i,k} &= \int_{I_i} dx \, |x_{\theta_k}\rangle \langle x_{\theta_k}|,
\end{align}
where $|x_{\theta_k}\rangle$ are the eigenstates of the quadrature operator $\hat{X}_{\theta_k}=(ae^{\mathrm{i}\theta_k}+a^\dag e^{-\mathrm{i}\theta_k})/\sqrt{2}$:
\begin{align}
\hat{X}_{\theta_k} |x_{\theta_k}\rangle &= x |x_{\theta_k}\rangle.
\end{align}



The set of measurement operators $\{\Pi_{i,k}\}$ forms a Positive Operator-Valued Measure (POVM) that describes the homodyne detection process under discretization. The POVM elements satisfy the normalization relation
\begin{align}
\sum_{k=0}^{N-1} \sum_{i=1}^{M} \Pi_{i,k} &= \mathbb{I},
\end{align}
where $\mathbb{I}$ is the identity operator in the Hilbert space of the signal field. This ensures that the probabilities of all possible measurement outcomes sum to unity, a defining property of any valid quantum measurement \cite{nielsen2010quantum, preskill1998lecture}. We further assume each phase is chosen with equal probability, then $\sum_{i=1}^{M} \Pi_{i,k} = \mathbb{I}/N$.

Each POVM element $\Pi_{i,k}$ is a positive semi-definite operator, and the probability of obtaining a measurement outcome corresponding to bin $I_i$ with phase $\theta_k$ is given by
\begin{align}
P(i,k) &= \text{Tr}(\rho \Pi_{i,k}).
\end{align}


To facilitate practical computations and theoretical analyses, it is often useful to express the POVM elements $\Pi_{i,k}$ in a specific basis. Here, we consider the photon number (Fock) basis and introduce a photon number truncation to make the calculations tractable. Assume that the Hilbert space is truncated to a maximum photon number $n_{\mathrm{max}}$, such that the basis states are $|0\rangle, |1\rangle, \ldots, |n_{\mathrm{max}}\rangle$.
Photon number truncation involves limiting the Hilbert space to a finite number of photon states. This is a common approximation in theoretical studies to simplify calculations, especially when dealing with states that have negligible probability amplitudes beyond a certain photon number.


The matrix elements of $\Pi_{i,k}$ in the photon number basis $\{|n\rangle\}$ are given by
\begin{equation}
\begin{aligned}
\Pi_{i,k}^{(m,n)} &= \langle m | \Pi_{i,k} | n \rangle \\
&= \int_{I_i} dx \, \langle m | x_{\theta_k} \rangle \langle x_{\theta_k} | n \rangle,
\end{aligned}
\end{equation}
where $m, n = 0, 1, \ldots, n_{\mathrm{max}}$.
The quadrature eigenstates $|x_{\theta_k}\rangle$ can be expressed in the photon number basis using the wavefunction
\begin{equation}
\begin{aligned}
\langle n | x_{\theta_k} \rangle &= \psi_n(x_{\theta_k}) \\
&= e^{\mathrm{i}n\theta}\frac{1}{\sqrt{2^n n! \sqrt{\pi}}} H_n(x_{\theta_k}) e^{-x_{\theta_k}^2 / 2},
\end{aligned}
\end{equation}
where $H_n(x)$ is a Hermite polynomial of order $n$.
Substituting this into the expression for $\Pi_{i,k}^{(m,n)}$, we obtain the matrix elements
\begin{equation}\label{eq:povm_metrix_element}
\begin{aligned}
\Pi_{i,k}^{(m,n)} &= \int_{I_i} dx \, \psi_m(x_{\theta_k}) \psi^*_n(x_{\theta_k}) \\
&= e^{\mathrm{i}(m-n)\theta_k}\int_{I_i} dx \, \frac{1}{\sqrt{2^m m! 2^n n! \pi}} H_m(x) H_n(x) e^{-x^2}.
\end{aligned}
\end{equation}
Across the paper, we use superscript $(m,n)$ to represent the label of matrix entries and subscript $i,k$ to represent the label of discrete phases and quadrature bins.







\section{Construct the Classical Shadow}

In this section, we propose the classical shadow of a continuous variable quantum state $\rho$ constructed by POVM of practical homodyne measurement, $\{\Pi_{i,k}\}$, where $i = 1, 2, \dots, M$ indexes the quadrature bins and $k = 1, 2, \dots, N$ indexes the discrete phases. The experimental settings are illustrated in Fig.~\ref{fig:overview}. First we define the linear map $C_E$ as
\begin{align}
    C_E(\rho) = \sum_{i=1}^M \sum_{k=1}^N \Tr(\rho \Pi_{i,k}) \Pi_{i,k}/|I_i|,
\end{align}
where $\Pi_{i,k}$ are the POVM elements, $|I_i| = x_{i+1}-x_i$ is the width of the quadrature bin, and $\Tr(\rho \Pi_{i,k})$ denotes the probability of obtaining a measurement outcome corresponding to the quadrature bin $I_i$ with phase $\theta_k$.

\begin{figure*}
    \centering
    \includegraphics[width=0.8\linewidth]{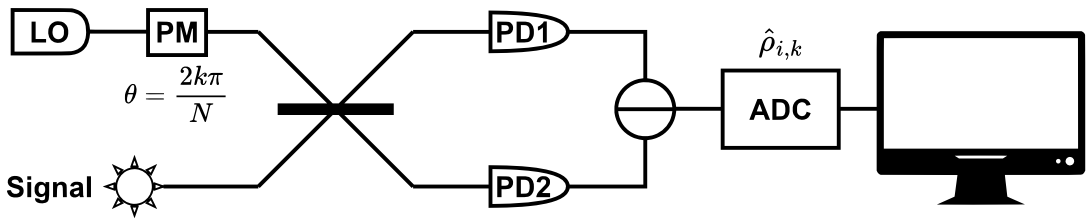}
    \caption{Overview of the experimental settings of our protocol. They are composed of a homodyne detector with a discrete-modulated local oscillator, an analog-to-digital converter discretizing the measurement output and a classical computer for the post-processing. LO: local oscillator; PM: phase modulator; PD: photo detector.}
    \label{fig:overview}
\end{figure*}
When the linear map $C_E$ is invertible, the classical shadow can be constructed as
\begin{align}\label{eq:def_shadow}
    \hat{\rho}_{i,k} = C_E^{-1}(\Pi_{i,k}/|I_i|),
\end{align}
where the outcome $(i,k)$ is obtained from the measurement. We will give the sufficient conditions and necessary conditions of an invertible $C_E$ in Sec.~\ref{sec:ICPOVM}.
We now formally prove that this construction yields an unbiased estimate of the quantum state $\rho$. The final estimator $\hat{\rho}$ is obtained by averaging a large number of single-shot estimators, or "snapshots", $\hat{\rho}_{i,k}$. By the linearity of expectation, it is sufficient to show that the expectation value of a single snapshot, taken over all possible measurement outcomes, is equal to the true state $\rho$.

The expectation value $\mathbb{E}[\hat{\rho}]$ is calculated by summing over all possible outcomes $(i,k)$, weighted by their respective probabilities of occurrence $P(i,k) = \text{Tr}(\rho\Pi_{i,k})$, 
\begin{align}
    \mathbb{E}[\hat{\rho}] &= \sum_{i=1}^{M} \sum_{k=1}^{N} P(i,k) \hat{\rho}_{i,k} \nonumber \\
    &= \sum_{i=1}^{M} \sum_{k=1}^{N} \text{Tr}(\rho\Pi_{i,k}) C_E^{-1}\left(\frac{\Pi_{i,k}}{|I_i|}\right),
\end{align}
where we have substituted the definition of $\hat{\rho}_{i,k}$ from Eq.~\eqref{eq:def_shadow}. Since the inverse map $C_E^{-1}$ is a linear superoperator, we can move it outside the summation:
\begin{align}
    \mathbb{E}[\hat{\rho}] &= C_E^{-1}\left( \sum_{i=1}^{M} \sum_{k=1}^{N} \text{Tr}(\rho\Pi_{i,k}) \frac{\Pi_{i,k}}{|I_i|} \right).
\end{align}
The expression inside the parentheses is precisely the definition of the linear map $C_E$ acting on the state $\rho$. Therefore, the equation simplifies to:
\begin{align}
    \mathbb{E}[\hat{\rho}] &= C_E^{-1}(C_E(\rho))  = \rho.
\end{align}
This confirms that each snapshot is an unbiased estimator of the quantum state. Consequently, the average over many snapshots, $\hat{\rho} = \frac{1}{T}\sum_{t=1}^T \hat{\rho}^{(t)}$, is also an unbiased estimator of $\rho$. This property is fundamental to the validity of the classical shadow formalism.

To estimate the expectation value of an observable $X$, we compute
\begin{equation}
\hat{X} = \Tr(X \hat{\rho}),
\end{equation}
where $\hat{\rho}$ is the classical shadow of the quantum state $\rho$. Then from the unbiasedness of classical shadow, $\hat{X}$ is also an unbiased estimator of $\Tr(X\rho)$.

The variance for estimating $\Tr(X\rho)$ is given by
\begin{equation}\label{eq:variance}
\begin{aligned}
    \mathrm{var}(\hat{X}) &=    
    \sum_{i=1}^M \sum_{k=1}^N \left| \Tr\left(C_E^{-1}(\Pi_{i,k}/|I_i|) X \right) \right|^2 \Tr(\rho \Pi_{i,k}) - \langle X \rangle^2 \\
    &\leq \lambda_{\mathrm{max}} \left\{ \sum_{i=1}^M \sum_{k=1}^N \left| \Tr\left(C_E^{-1}(\Pi_{i,k}/|I_i|) X \right) \right|^2 \Pi_{i,k} \right\} \\
    &:= \| X \|_E,
\end{aligned}
\end{equation}
where $\lambda_{\mathrm{max}}(\cdot)$ denotes the maximal eigenvalue of the corresponding operator. The term $\| X \|_E$ is referred to as the shadow norm of the observable $X$.

From Bernstein's inequality, we have
\begin{align}
    \mathbf{P}\left\{
        |\hat{X} - \langle X \rangle| \geq \epsilon
    \right\}
    \leq
    2 \exp\left( \frac{-\epsilon^2 / 2}{\| X \|_E + 2\epsilon / 3} \right).
\end{align}
If we use the shadow obtained by performing $T$ measurements to estimate $X$, then the corresponding probability bound is
\begin{align}
    \mathbf{P}\left\{
        |\hat{X} - \langle X \rangle| \geq \epsilon
    \right\}
    \leq
    2 \exp\left( \frac{-T\epsilon^2 / 2}{\| X \|_E + 2\epsilon / 3} \right).
\end{align}

In summary, we construct the classical shadow of a quantum state $\rho$ using a set of POVMs $\{\Pi_{i,k}\}$ corresponding to different quadrature measurements. The linear map $C_E$ aggregates the measurement data, and its inverse is used to create unbiased estimators $\hat{\rho}_{i,k}$ for the density matrix elements. By averaging these estimators, we obtain an estimate $\hat{\rho}$ of the quantum state. The variance of the observable estimate $\hat{X}$ is bounded by the shadow norm $\| X \|_E$, which quantifies the sensitivity of the observable to the measurement scheme. Bernstein's inequality provides probabilistic guarantees that the estimate $\hat{X}$ is close to the true expectation value $\langle X \rangle$ with high probability, depending on the number of measurements $T$ and the shadow norm. This framework is crucial for efficiently estimating expectations from experimental data, ensuring both accuracy and reliability in quantum information processing tasks.

\section{Conditions for Informationally Complete POVMs}\label{sec:ICPOVM}
The existence of inverse map $C_E^{-1}$ in Eq.~\eqref{eq:def_shadow} requires that the POVM $\{\Pi_{i,k}\}$ is informationally complete \cite{busch1995operational, d2004quantum, scott2006tight}. In this section, we give the sufficient and necessary conditions for it.

\begin{theorem}
    Given a photon number cutoff $n_{\mathrm{max}}$, when
    \begin{equation}
    \begin{aligned}
    N &\geq 2n_{\mathrm{max}} + 1, \\
    M &\geq n_{\mathrm{max}} + 1,
    \end{aligned}
    \end{equation}
    there exists a set of quadrature bins $\{I_i\}$ such that the POVM $\{\Pi_{i,k}\}$, where $i = 1, 2, \dots, M$ and $k = 1, 2, \dots, N$, is informationally complete.
    \label{theorem:1}
\end{theorem}
This theorem establishes that with a sufficiently large number of phase settings and quadrature bins, specifically when $N \geq 2n_{\mathrm{max}} + 1$ and $M \geq n_{\mathrm{max}} + 1$, it is possible to design a homodyne detection scheme that fully captures all the necessary information to reconstruct any quantum state truncated at $n_{\mathrm{max}}$ photons. 

\begin{proof}
We give a proof by contradiction to demonstrate the informationally completeness of the POVM $\{\Pi_{i,k}\}$. The method involves assuming the POVM is incomplete and showing this leads to a contradiction given the conditions on the number of discrete phases $N$ and quadrature bins $M$, thus proving it captures all state information.

Each POVM element $\Pi_{i,k}$ corresponds to a projection onto a specific quadrature bin $I_i$ at a particular phase setting $\theta_k$. We express these elements in the photon number (Fock) basis, considering a photon number truncation up to $n_{\mathrm{max}}$. This truncation limits the Hilbert space to a finite dimension of $(n_{\mathrm{max}} + 1)$, facilitating the analysis.
Recall Eq.~\eqref{eq:povm_metrix_element}, the matrix elements of $\Pi_{i,k}$ in the photon number basis $\{|n\rangle\}$ are given by an integration of the Hermite polynomials. For computational purposes, each POVM element $\Pi_{i,k}$ is vectorized by stacking its columns into a single column vector $\Vec{\Pi}_{i,k}$. This results in a vector of length $(n_{\mathrm{max}} + 1)^2$. Collecting all POVM elements, we construct the measurement matrix $E$ as
    \begin{align}
        E = 
        \begin{pmatrix}
            \Vec{\Pi}_{1,1} & \Vec{\Pi}_{1,2} & \cdots & \Vec{\Pi}_{M,N}
        \end{pmatrix},
    \end{align}
    where $E$ is a matrix of size $(n_{\mathrm{max}} + 1)^2 \times (M \times N)$.

    Assume, for contradiction, that the POVM $\{\Pi_{i,k}\}$ is not informationally complete. This implies that the measurement matrix $E$ does not have full rank, i.e.,
    \begin{align}
        \text{rank}(E) &< (n_{\mathrm{max}} + 1)^2.
    \end{align}
    Therefore, there exists a non-trivial vector $\Vec{r} \in \mathbb{C}^{(n_{\mathrm{max}} + 1)^2}$ such that
    \begin{align}
       \Vec{r} E  &= 0.
    \end{align}
    Expanding this equation, we have
    \begin{align}
        \sum_{i=1}^{M} \sum_{k=1}^{N} r_{i,k} \Vec{\Pi}_{i,k} &= 0,
    \end{align}
    where not all coefficients $r_{i,k}$ are zero.

    To analyze this equation, we express the measurement matrix $E$ in terms of phase contributions corresponding to $e^{\mathrm{i}(m-n)\theta}$ and Hermite polynomial contributions corresponding to the rest integration in Eq.~\eqref{eq:povm_metrix_element}. Define the following vectors:
    \begin{equation}\label{eq:phase_coeff}
    \begin{aligned}
        \Vec{C}_{a} &= 
        \begin{pmatrix}
            1, & \cos\left(\frac{2a\pi}{N}\right), & \cdots, & \cos\left(\frac{2a(N-1)\pi}{N}\right)
        \end{pmatrix}, \\
        \Vec{S}_{a} &= 
        \begin{pmatrix}
            0, & \sin\left(\frac{2a\pi}{N}\right), & \cdots, & \sin\left(\frac{2a(N-1)\pi}{N}\right)
        \end{pmatrix},
    \end{aligned}
    \end{equation}
    for $a = 0, 1, \ldots, n_{\mathrm{max}}$.
    Additionally, define
    \begin{align}\label{eq:vec_Hmn}
        \Vec{H}^{(m,n)} &= 
        \begin{pmatrix}
            H^{(m,n)}_1, & H^{(m,n)}_2, & \cdots, & H^{(m,n)}_M
        \end{pmatrix},
    \end{align}
    where
    \begin{align}
        H_i^{(m,n)} &= \frac{1}{\sqrt{2^{m+n} m! n! \pi}} \int_{x_i}^{x_{i+1}} H_m(x_{\theta_i}) H_n(x_{\theta_i}) e^{-x_{\theta_i}^2} dx.
    \end{align}
 With these definitions, each row vector $\vec{E}^{(m,n)}$ of the measurement matrix $E$ can be expressed as
    \begin{align}
        \vec{E}^{(m,n)} &= \Vec{C}_{m-n} \otimes \Vec{H}^{(m,n)} + \mathrm{i} \Vec{S}_{m-n}\otimes \Vec{H}^{(m,n)},
    \end{align}
    where $\otimes$ denotes the Kronecker product.
    Performing row operations that preserve the rank of $E$, we transform $E$ into a new matrix $E'$:
    \begin{align}
        E' = 
        \begin{pmatrix}  
            \Vec{E}^{\prime(0,0)} \\  
            \Vec{E}^{\prime(0,1)}  \\  
            \vdots \\  
            \Vec{E}^{\prime(n_{\mathrm{max}},n_{\mathrm{max}})}  
        \end{pmatrix},
    \end{align}
    where each transformed row $\Vec{E}^{\prime (m,n)}$ is given by
    \begin{align}
        \Vec{E}^{\prime (m,n)} =
        \begin{cases}  
            \Vec{C}_{m-n} \otimes \Vec{H}^{(m,n)} & \text{if } m - n \leq 0, \\
            \Vec{S}_{m-n} \otimes \Vec{H}^{(m,n)} & \text{if } m - n > 0.
        \end{cases}
    \end{align}
    
    Under the assumption that $\text{rank}(E') < (n_{\mathrm{max}} + 1)^2$, there exists a non-trivial set of coefficients $r_{m,n}$ such that
    \begin{align}
        0 &= \sum_{m,n} r_{m,n} \Vec{E}^{\prime (m,n)} \nonumber \\
        &= \sum_{m \leq n} r_{m,n} \Vec{C}_{m-n} \otimes \Vec{H}^{(m,n)} + \sum_{m > n} r_{m,n} \Vec{S}_{m-n} \otimes \Vec{H}^{(m,n)}.
    \end{align} 
    Rewriting the above equation, we obtain
    \begin{align}
        0 &= \sum_{a=0}^{n_{\mathrm{max}}} \Vec{C}^T_{a} \sum_{m-n = -a} r_{m,n} \Vec{H}^{(m,n)} + \sum_{a=1}^{n_{\mathrm{max}}} \Vec{S}^T_{a} \sum_{m-n = a} r_{m,n} \Vec{H}^{(m,n)} \nonumber \\
        &= AB,
    \end{align}
    where
    \begin{align}\label{eq:def_matrixA}
        A &= 
        \begin{pmatrix}
            \Vec{C}_{0}^T & \Vec{C}_{1}^T & \Vec{S}_{1}^T & \Vec{C}_{2}^T & \Vec{S}_{2}^T & \cdots & \Vec{C}_{n_{\mathrm{max}}}^T & \Vec{S}_{n_{\mathrm{max}}}^T
        \end{pmatrix},
    \end{align}
    and
    \begin{align}
        B &=   
        \begin{pmatrix}  
            \sum_{m-n = 0} r_{m,n} \Vec{H}^{(m,n)} \\  
            \sum_{m-n = -1} r_{m,n} \Vec{H}^{(m,n)}\\  
            \sum_{m-n = 1} r_{m,n} \Vec{H}^{(m,n)} \\  
            \vdots \\  
            \sum_{m-n = -n_{\mathrm{max}}} r_{m,n} \Vec{H}^{(m,n)} \\
            \sum_{m-n = n_{\mathrm{max}}} r_{m,n} \Vec{H}^{(m,n)} \\  
        \end{pmatrix}.
    \end{align}
    Therefore, if the POVM $\{\Pi_{i,k}\}$ is not informationally complete, there exist matrices $A$ and $B$ such that
    \begin{align}
        AB = 0.
    \end{align}
    
    However, under the condition $N \geq 2n_{\mathrm{max}} + 1$, the following Lemma~\ref{lemma:1} ensures that the matrix $A$ is full rank, and thus the only solution to $AB = 0$ is $B = 0$.
 With condition $M \geq n_{\mathrm{max}} + 1$, Lemma~\ref{lemma:2} ensures that all coefficients $r_{m,n} = 0$, which contradicts our initial assumption that the vector $\Vec{r}$ is non-trivial.  
\end{proof}

\begin{lemma}
    \label{lemma:1}
    Given a photon number cutoff $n_{\mathrm{max}}$ and a matrix $A$ defined in Eq.~\eqref{eq:def_matrixA}, if the number of phase settings satisfies
    \begin{align}
        N &\geq 2n_{\mathrm{max}} + 1,
    \end{align}
    then the matrix $A$ is column-full rank.
\end{lemma}

\begin{proof}
    Suppose, for the sake of contradiction, that the matrix $A$ is not of full column rank. This implies that there exists a non-trivial set of coefficients $\{c_a\}$ and $\{s_a\}$ such that
    \begin{align}\label{eq:zero_vector}
        c_0 \Vec{C}_0 + \sum_{j=1}^{n_{\mathrm{max}}} \left( c_a \Vec{C}_a + s_a \Vec{S}_a \right) = 0.
    \end{align}
    Here, $\Vec{C}_a$ and $\Vec{S}_a$ are vectors associated with the cosine and sine components of the phase settings defined in Eq.~\eqref{eq:phase_coeff}, respectively.
    
    Define the function $f(\theta)$ as
    \begin{align}
        f(\theta) = c_0 + \sum_{a=1}^{n_{\mathrm{max}}} \left( c_a \cos(a\theta) + s_a \sin(a\theta) \right).
    \end{align}
    Then the zero vector in Eq.~\eqref{eq:zero_vector}  implies that $f(\theta) = 0$ for all phases  $\theta_k = \frac{2k\pi}{N}$, where $k =0, 1, 2, \dots, N-1$.
    
    A trigonometric polynomial of degree $n_{\mathrm{max}}$ can have at most $2n_{\mathrm{max}}$ zeros in the interval $[0, 2\pi)$. However, our assumption leads to $N \geq 2n_{\mathrm{max}} + 1$ zeros, which exceeds this maximum number. This contradiction implies that our initial assumption is false.
    Therefore, the matrix $A$ must be column-full rank.
\end{proof}

Continuing the inference, if $AB=0$, then $B$ must be a zero matrix. However, we can always find a suitable set of quadrature bins $\{I_i\}$ such that $B \neq 0$. We summarize this as the following lemma.

\begin{lemma}
    \label{lemma:2}
    Given a photon number cutoff $n_{\mathrm{max}}$, 
    if the number of quadrature bins satisfies
    \begin{align}
        M &\geq n_{\mathrm{max}} + 1,
    \end{align}
    then there exists a set of quadrature bins $\{I_i\}$ such that $B \neq 0$ unless all coefficients $r_{m,n}=0$.
\end{lemma}

\begin{proof}
To prove this lemma, we define the matrix $H_a$ for each $a = 0, 1, \dots, n_{\max}$ as
\begin{align}
H_a = \big( \vec{H}^{(0,a)} \ \vec{H}^{(1,a+1)} \ \dotsc \ \vec{H}^{(n_{\max}-a,n_{\max})} \big),
\end{align}
where $\vec{H}^{(m,n)}$ are vectors derived from the integrals of Hermite polynomials over the quadrature bins in Eq.~\eqref{eq:vec_Hmn}. Then the matrix $B$ is composed of linear combinations of different $H_a$. The idea is to make all $H_a$ be column-full rank.

We start with an initial bin $I_1 = [x_1, x_2)$. This determines the first row of each $H_a$, which is denoted as $\vec{H}_a^{(1)}$. Then 
\begin{align}
         \vec{H}_{a}^{(1)} \cdot \vec{q}_{a,e} = 0
     \end{align}
 defines a $(n_{\mathrm{max}} - a)$-dimensional null space spanned by a set of basis $\{\vec{q}_{a,e}\}$ with $e \in \{ 1, 2, \dots, n_{\mathrm{max}} - a\}$.
 Next, define the function
    \begin{align}
       F(x_3, e, a) = \vec{H}^{(2)}_{a} \cdot \vec{q}_{a,e},
    \end{align}
where $\vec{H}^{(2)}_{a}$ represents the second row of $H_a$.
Each function $F$ is a polynomial in terms of $x_3$ when the initial point of the second quadrature bin $x_2$ is already given. For a given $a$, if $F(x_3, e, a) = 0$ for all $e$, then the first two rows are linearly dependent. By the fundamental theorem of algebra, each $F(x_3, e, a)$ can have at most $2n_{\mathrm{max}} +1$ zero points for $x_3$. Therefore, there always exists $x_3$ such that $F(x_3, e, a) \neq 0$. The second row contributes to increasing the column rank of by one. We can repeat such process until $H_a$ is column-full rank.

Next we discuss the number of necessary quadrature bins. Notice that the number of columns of $H_a$ is $n_{\max} - a + 1$. When $a = 0$, the integral $H_n^2(x)e^{-x^2}$ are strictly positive and linearly independent. This means there is no inherent linear dependency among the rows of $H_0$, allowing it to achieve full column rank when $M \geq n_{\max} + 1$. On the other hand, when $a \geq 1$, the orthogonality of Hermite polynomials implies that $\int_{-\infty}^\infty H_n(x)H_{n+a}(x)e^{-x^2}dx = 0$. As a result, for any choice of bins, the sum of all rows of $H_a$ must be zero. This imposes a single linear constraint on the rows, meaning that the number of quadrature bins should satisfy $M \geq n_{\max} - a + 2$. In this case, we have a sufficient condition $M \geq n_{\max} + 1$ corresponding to the minimum value $a=1$.

The matrix $B$ constructed from these $H_a$ cannot be a zero matrix unless all coefficients vanish. Therefore, as long as $M \geq n_{\mathrm{max}} + 1$, there exists a set of quadrature bins $\{I_i\}$ such that $B = 0$ implies all coefficients $r_{m,n}=0$, completing the proof.
\end{proof}

We design an algorithm to set the quadrature bins in Algorithm~\ref{alg:adaptive_quad_bins} to generate equal-spaced quadrature bins guaranteeing informational completeness. We note that the edge points are chosen according to the experimental data collected, i.e., $x_1$ should satisfy $x_1 \leq x_\mathrm{min}$ where $x_\mathrm{min}$ is the minimum quadrature value detected. Similarly, $x_{M+1}\geq x_\mathrm{\max}$ should also be satisfied. We let $L_0 = \mathrm{max}(|x_{\mathrm{min}}|,|x_{\mathrm{max}}|)$ be an input of Algorithm~\ref{alg:adaptive_quad_bins} to characterize the range of measurement outputs. The algorithm can find a valid set of quadrature bins $\{I_i\}$ with high probability because, under the conditions $M \geq n_{\max} + 1$ and $N \geq 2n_{\max} + 1$, the measurement matrix $E$ depends analytically on the bin edges through integrals of Hermite polynomials. Since the set of bin configurations that cause $E$ to be rank-deficient corresponds to zeros of nontrivial analytic functions, it forms a measure-zero subset in the space of all possible bin choices. By starting from an initial range $L_0 = \max(|x_{\min}|, |x_{\max}|)$ and gradually expanding $L$, Algorithm~\ref{alg:adaptive_quad_bins} explores a sequence of uniformly spaced bin partitions over increasingly large intervals, thereby almost surely avoiding these zero points. Consequently, the algorithm eventually encounters a binning for which $\operatorname{rank}(E) = (n_{\max}+1)^2$, ensuring informational completeness with high probability in practice.








\begin{algorithm}[htbp]
\caption{Generating equal-spaced quadrature bins with informational completeness verification}\label{alg:adaptive_quad_bins}
\SetAlgoLined 
\KwIn{
$n_{\mathrm{max}}$ (photon number cutoff); \\
$N$ (number of discrete phases, $N \geq 2n_{\mathrm{max}}+1$) \\
$M$ (number of quadrature bins, $M \geq n_{\mathrm{max}}+1$); \\
$L_0$ (initial quadrature range); \\
$\delta L$ (range expansion increment);
}
\KwOut{
Set of quadrature bins $\{I_i\}$ where POVM $\{\Pi_{i,k}\}$ is informationally complete
}

\textbf{Initialize:} \\
$L \gets L_0$; $\text{isComplete} \gets \text{false}$;

\While {$\neg \text{isComplete}$}{
    $\Delta I \gets \dfrac{2 \cdot L}{M}$; \\ 
    \For{$i \gets 1$ \KwTo $M+1$}{
        $x_i \gets -L + \Delta I \cdot (i-1)$; 
    }
    Define bins: $I_i = [x_i, x_{i+1})$ for $i = 1, 2, \dots, M$; \\
    
    Initialize $E$ as $\mathbb{C}^{(n_{\mathrm{max}}+1)^2 \times (M \times N)}$ matrix; \\
    \For{$k = 0, 1, \dots, N-1$}{
        \For{$i \gets 1$ \KwTo $M$}{
             Compute POVM element $\Pi_{i,k}$ via Eq.~\eqref{eq:povm_metrix_element}\;
            Vectorize $\Pi_{i,k}$ into $\vec{\Pi}_{i,k} \in \mathbb{C}^{(n_{\mathrm{max}}+1)^2}$; 
            Assign $\vec{\Pi}_{i,k}$ to $(k \cdot M + i)$-th column of $E$;
        }
    }
    Form $E = \big( \vec{\Pi}_{1,0}, \vec{\Pi}_{2,0}, \dots, \vec{\Pi}_{M,N-1} \big)$; \\
    
    $r \gets \mathrm{rank}(E)$;\\
    \If{$r == (n_{\mathrm{max}}+1)^2$}{
        $\text{isComplete} \gets \text{true}$; \\
        \Return $\{I_i\}$;
    }
        $L \gets L + \delta L$; 
    
}
\end{algorithm}

\begin{theorem}\label{theorem:2}
    Given a photon number cutoff $n_{\mathrm{max}}$, the POVMs $\{\Pi_{i,k}\}$ are informationally complete only if the number of phase settings $N$ satisfies one of the following conditions
        \begin{enumerate}
        \item $N \geq 2n_{\mathrm{max}} + 1$,
        \item $n_{\mathrm{max}}  < N \leq 2n_{\mathrm{max}}, \quad N \; \text{is odd}$.
        \end{enumerate}
    \end{theorem}
This theorem delineates the necessary conditions for the POVM $\{\Pi_{i,k}\}$ to achieve informational completeness. It asserts that the informationally complete POVM is only attainable if the number of discrete phases $N$ is no less than $2n_{\mathrm{max}} + 1$, or if it lies within the range $n_{\mathrm{max}} < N \leq 2n_{\mathrm{max}} $ with $N$ being an odd integer. Intuitively, this means that there is a critical balance between the number of phase settings and the photon number cutoff that must be maintained to ensure that the measurements are sufficiently diverse and finely resolved to capture all aspects of the quantum state.

\begin{proof}
  To intuitively show the POVM $\{\Pi_{i,k}\}$ is not informationally complete under certain conditions on $N$, the approach is to construct distinct quantum states that yield identical measurement probabilities, proving the POVM cannot distinguish between them. 
Assume that the phase settings are uniformly distributed, such that        $\theta_k = 2\pi k/N$, for $k = 1, 2, \dots, N$. First, if $N \leq n_{\max}$, we will prove that POVM $\{\Pi_{i,k}\}$ cannot fully reconstruct the quantum state.

    Since $N \leq n_{\max}$, we can always choose a quantum state
    \begin{align}
        |\phi\rangle = \frac{\alpha|0\rangle + \beta |N\rangle}{\sqrt{2}},
    \end{align}
    where $\alpha$ is a real coefficient and $\beta$ is a complex coefficient. This state is a superposition of the vacuum state $|0\rangle$ and the $N$-photon Fock state $|N\rangle$. 
The probability of obtaining outcome $(i,k)$ when measuring the quantum state $|\phi\rangle$ is
    \begin{equation}\label{eq:p_ik}
    \begin{aligned}
        P(i,k) &= \text{Tr}(\Pi_{i,k} |\phi\rangle\langle\phi|) \\
               &= \frac{1}{2} \left( \Pi_{i,k}^{(0,0)} \alpha^2 + \Pi_{i,k}^{(N,N)} |\beta|^2 + \Pi_{i,k}^{(0,N)} \alpha \beta +\Pi_{i,k}^{(N,0)}\alpha \beta^*  \right),
    \end{aligned}
    \end{equation}
which involves the following POVM matrix elements,
    \begin{equation}
    \begin{aligned}
        \Pi_{i,k}^{(0,0)} &= \frac{1}{\sqrt{\pi}} \int_{x_i}^{x_{i+1}} e^{-x^2} H_0(x) H_0(x) dx, \\
        \Pi_{i,k}^{(0,N)} &= \frac{e^{-\mathrm{i}N\theta_k}}{\sqrt{2^{N} N! \pi}} \int_{x_i}^{x_{i+1}} e^{-x^2} H_0(x) H_{N}(x) dx, \\
        \Pi_{i,k}^{(N,0)} &= \frac{e^{\mathrm{i}N\theta_k}}{\sqrt{2^{N} N! \pi}} \int_{x_i}^{x_{i+1}} e^{-x^2} H_{N}(x) H_0(x) dx, \\
        \Pi_{i,k}^{(N,N)} &= \frac{1}{\sqrt{2^{2N} (N!)^2 \pi}} \int_{x_i}^{x_{i+1}} e^{-x^2} H_{N}(x) H_N(x) dx.
    \end{aligned}
    \end{equation}
Notice that the global phases $e^{-\mathrm{i}N\theta_k}=e^{\mathrm{i}N\theta_k}=1$. Then $ \Pi_{i,k}^{(0,N)} =  \Pi_{i,k}^{(N,0)}$ and Eq.~\eqref{eq:p_ik} is simplified to  
    \begin{equation}
    \begin{aligned}
        P(i,k)= \frac{1}{2} \left( \Pi_{i,k}^{(0,0)} \alpha^2 + \Pi_{i,k}^{(N,N)} |\beta|^2 + 2 \Pi_{i,k}^{(0,N)} \alpha \text{Re}(\beta) \right).
    \end{aligned}
    \end{equation}
    This probability expression depends only on the real part of $\beta$ and its magnitude $|\beta|$. i.e., the measurement probabilities is independent of the imaginary part. This implies that the POVM $\{\Pi_{i,k}\}$ cannot distinguish $\ket{\phi} = (\alpha|0\rangle + \beta |N\rangle)/\sqrt{2}$ and $\ket{\phi^\prime} = (\alpha|0\rangle + \beta^* |N\rangle)/\sqrt{2}$, indicating that information about the imaginary component of $\beta$ is inaccessible through these measurements. Therefore, the POVM is incomplete.
    
    Next, we assume that the number of phase settings $N$ satisfies
    \begin{align}
        n_{\mathrm{max}} < N \leq 2n_{\mathrm{max}}  \quad \text{and} \quad N \quad \text{is even}.
    \end{align}
  Similarly, we consider a quantum state
    \begin{align}
        |\psi\rangle = \frac{\alpha|0\rangle + \beta|\frac{N}{2}\rangle}{\sqrt{2}},
    \end{align}
   where $\alpha$ is real and $\beta$ is complex.
   The probability of obtaining outcome $(i,k)$ when measuring the quantum state $|\psi\rangle$ is
 \begin{equation}  
 \begin{aligned}
        P(i,k) &= \text{Tr}(\Pi_{i,k} |\psi\rangle\langle\psi|) \\
               &= \frac{1}{2} \left( \Pi_{i,k}^{(0,0)} \alpha^2 + \Pi_{(i,k)}^{(\frac{N}{2},\frac{N}{2})} |\beta|^2 - 2 \Pi_{i,k}^{(0,\frac{N}{2})} \alpha \text{Re}(\beta) \right), 
    \end{aligned}
    \end{equation}
   where the POVM matrix elements are
   \begin{equation}
    \begin{aligned}
        \Pi_{i,k}^{(0,0)} &= \frac{1}{\sqrt{\pi}} \int_{x_i}^{x_{i+1}} e^{-x^2} H_0(x) H_0(x) dx, \\
        \Pi_{i,k}^{(0,\frac{N}{2})} &= -\frac{1}{\sqrt{2^{\frac{N}{2}} (\frac{N}{2})! \pi}} \int_{x_i}^{x_{i+1}} e^{-x^2} H_0(x) H_{\frac{N}{2}}(x) dx, \\
        \Pi_{i,k}^{(\frac{N}{2},0)} &= -\frac{1}{\sqrt{2^{\frac{N}{2}} (\frac{N}{2})! \pi}} \int_{x_i}^{x_{i+1}} e^{-x^2} H_{\frac{N}{2}}(x) H_0(x) dx, \\
        \Pi_{i,k}^{(\frac{N}{2},\frac{N}{2})} &= \frac{1}{\sqrt{2^{N} (\frac{N}{2})! (\frac{N}{2})! \pi}} \int_{x_i}^{x_{i+1}} e^{-x^2} H_{\frac{N}{2}}(x) H_{\frac{N}{2}}(x) dx.
    \end{aligned}
    \end{equation}
 Similar to the previous case, this probability expression depends only on $\text{Re}(\beta)$ and $|\beta|$, and is insensitive to the sign of $\text{Im}(\beta)$. This indicates that even when $N$ is within the range $n_{\mathrm{max}} < N \leq 2n_{\mathrm{max}} $ and $N$ is even, the POVM $\{\Pi_{i,k}\}$ remains incomplete as it cannot distinguish $\ket{\psi}$ and $\ket{\psi^\prime} = (\alpha|0\rangle + \beta^*|\frac{N}{2}\rangle)/\sqrt{2}$.
    
    Therefore, under both scenarios, when $N \leq n_{\mathrm{max}}$, and when $n_{\mathrm{max}} < N \leq 2n_{\mathrm{max}}$ with an even $N$, the POVM $\{\Pi_{i,k}\}$ fails to be informationally complete. This concludes the proof of Theorem \ref{theorem:2}.
\end{proof}

Overall, Theorem \ref{theorem:1} and Theorem \ref{theorem:2} collectively provide a comprehensive framework for understanding the requirements of informationally completeness in homodyne detection-based shadow estimation. Theorem \ref{theorem:1} offers sufficient conditions, ensuring that with enough phase settings and quadrature bins, complete state reconstruction is achievable. Conversely, Theorem \ref{theorem:2} outlines the boundaries within which informationally completeness is possible, highlighting the necessity of a minimum and appropriately constrained number of phase settings relative to the photon number cutoff. Physically, these theorems underscore the importance of optimizing the measurement setup by balancing the diversity of phase angles and the resolution of quadrature bins. Experimentally, they serve as crucial guidelines for designing homodyne detection systems, ensuring that the chosen number of phases and quadrature bins are adequate for accurate and complete quantum state characterization. This is particularly vital for applications involving high-precision quantum information protocols, where the fidelity of state reconstruction directly impacts the performance and reliability of quantum technologies.

\section{Variance Analysis}
The performance of a classical shadow protocol is ultimately quantified by the precision of its observable estimations—specifically, the variance of the estimator for an arbitrary observable $X$, which directly determines the number of measurements required to achieve a target accuracy.
In this section, we give the upper bound on the variance of shadow estimation in terms of key system parameters.

\begin{theorem}\label{theorem:3}
Considering an observable $X$, the variance of its estimator $\hat{X}$ has an upper bound in terms of system parameters $N$, $M$ and $n_{\mathrm{max}}$,
\begin{align}
\mathrm{var}(\hat{X}) \leq N (n_{\mathrm{max}}+1) M^2 ||X||_\infty^2.
\end{align}
\end{theorem}

This theorem establishes a fundamental limit on the precision of measuring the operator $X$ in a CV quantum system with practical homodyne detection. 
Note that, according to Theorem~\ref{theorem:1}, for the POVM to be informationally complete, $N$ and $M$ need to be $\mathcal{O}(n_{\mathrm{max}})$. Therefore, we can relate the variance upper bound to $n_{\mathrm{max}}$, i.e., $\mathrm{var}(\hat{X}) \sim \mathcal{O}(n_{\mathrm{max}}^4)$.

\begin{proof}
Recalling Eq.~\eqref{eq:variance}, we define $\| X \|_E$ as an upper bound of $\text{var}(\hat{X})$. Next we try to find an upper bound of $\| X \|_E$ in terms of system parameters, i.e., the number of quadrature bins $M$, the number of phase angle samples $N$, and the photon number cutoff $n_\mathrm{max}$. We introduce $L = \sum_{i=1}^M |I_i|$ is the total length of the quadrature bins. We assume each phase of local oscillator is chosen with equal probability in homodyne detection. Then $\sum_{i=1}^M \Pi_{i,k} = \frac{1}{N} \mathbb{I}$ for each fixed $k$. 

If we consider the vectorization of density matrix $\ket{\rho}$, the map $C_E(\rho)$ has an operator form $C_E$.
For any density matrix $\rho$ with $\|\rho\|_{\text{HS}} = 1$, where $\|\cdot\|_{\text{HS}}$ denotes the Hilbert-Schmidt norm, we have,
\begin{align}
\langle \rho | C_E | \rho \rangle = \sum_{i=1}^M \sum_{k=1}^N \frac{[\text{Tr}(\rho \Pi_{i,k})]^2}{|I_i|}.
\end{align}
Using the Cauchy-Schwarz inequality for each fixed $k$, we obtain
\begin{align}
\sum_{i=1}^M \frac{[\text{Tr}(\rho \Pi_{i,k})]^2}{|I_i|} \geq \frac{\left( \sum_{i=1}^M \text{Tr}(\rho \Pi_{i,k}) \right)^2}{\sum_{i=1}^M |I_i|} = \frac{\left( \frac{1}{N} \right)^2}{L} = \frac{1}{N^2 L}.
\end{align}
Summing over all $k$, we have
\begin{align}
\langle \rho | C_E | \rho \rangle \geq \sum_{k=1}^N \frac{1}{N^2 L} = \frac{N}{N^2 L} = \frac{1}{N L}.
\end{align}
Thus, the minimum eigenvalue of $C_E$ satisfies
\begin{align}
\lambda_{\text{min}} \geq \frac{1}{N L}.
\end{align}

Next, consider the operator inside the definition of $\| X \|_E$,
\begin{align}
\tilde{X} = \sum_{i=1}^M \sum_{k=1}^N \left| \Tr\left(C_E^{-1}(\Pi_{i,k}/|I_i|) X \right) \right|^2 \Pi_{i,k}.
\end{align}
We wish to bound the term $|\text{Tr}(C_E^{-1}(\Pi_{i,k}/|I_i|) X)|^2$. Let $\hat{\rho}_{i,k} = C_E^{-1}(\Pi_{i,k}/|I_i|)$ denote the single shadow snapshot.
First, we apply the Cauchy-Schwarz inequality for the Hilbert-Schmidt inner product ($|\text{Tr}(A^\dagger B)| \le ||A||_{\mathrm{HS}} ||B||_{\mathrm{HS}}$), noting that $\hat{\rho}_{i,k}$ is Hermitian since $C_E$ and $\Pi_{i,k}$ are self-adjoint:
\begin{align}
    |\text{Tr}(\hat{\rho}_{i,k} X)| \le ||\hat{\rho}_{i,k}||_{\mathrm{HS}} ||X||_{\mathrm{HS}}.
\end{align}
Squaring both sides yields:
\begin{align}
    |\text{Tr}(\hat{\rho}_{i,k} X)|^2 \le ||\hat{\rho}_{i,k}||_{\mathrm{HS}}^2 ||X||_{\mathrm{HS}}^2.
    \label{eq:bound_hs_step1}
\end{align}
Next, we bound the HS norm of the snapshot using the induced HS norm of the superoperator $C_E^{-1}$. Since $C_E$ is a positive self-adjoint superoperator, the induced HS norm of its inverse is equal to the inverse of its minimum eigenvalue in the HS sense: $||C_E^{-1}||_{\mathrm{HS} \to \mathrm{HS}} = 1/\lambda_{\mathrm{min}}^{(\mathrm{HS})}(C_E)$.
\begin{equation}
\begin{aligned}
    ||\hat{\rho}_{i,k}||_{\mathrm{HS}} &= ||C_E^{-1}(\Pi_{i,k}/|I_i|)||_{\mathrm{HS}} \\
    &\le ||C_E^{-1}||_{\mathrm{HS} \to \mathrm{HS}} \cdot ||\Pi_{i,k}/|I_i|||_{\mathrm{HS}} \\
    &= \frac{1}{\lambda_{\mathrm{min}}^{(\mathrm{HS})}(C_E)} \frac{||\Pi_{i,k}||_{\mathrm{HS}}}{|I_i|}.
\end{aligned}
\end{equation}
Squaring this gives:
\begin{align}
    ||\hat{\rho}_{i,k}||_{\mathrm{HS}}^2 \le \frac{1}{(\lambda_{\mathrm{min}}^{(\mathrm{HS})}(C_E))^2} \frac{||\Pi_{i,k}||_{\mathrm{HS}}^2}{|I_i|^2}.
    \label{eq:bound_hs_step2}
\end{align}
Substituting Eq.~\eqref{eq:bound_hs_step2} into Eq.~\eqref{eq:bound_hs_step1}, we have
\begin{align}
    |\text{Tr}(\hat{\rho}(i,k) X)|^2 \le \frac{||\Pi_{i,k}||_{\mathrm{HS}}^2}{(\lambda_{\mathrm{min}}^{(\mathrm{HS})}(C_E))^2 |I_i|^2} \cdot ||X||_{\mathrm{HS}}^2.
    \label{eq:bound_hs_step3}
\end{align}
Finally, we relate the HS norm of $X$ to its spectral norm (operator norm) $||X||_\infty$. For an operator $X$ acting on an $(n_{\mathrm{max}}+1)$-dimensional space, $||X||_{\mathrm{HS}}^2 = \sum_j s^2_j(X) \le (n_{\mathrm{max}}+1) s^2_1(X) = (n_{\mathrm{max}}+1) ||X||_\infty^2$, where $s_j(X)$ is the $j$-th largest singular value of $X$. Applying this we have
\begin{equation}
\begin{aligned}
    |\text{Tr}(C_E^{-1}(\Pi_{i,k}/|I_i|) X)|^2 &\le \frac{||\Pi_{i,k}||_{\mathrm{HS}}^2}{(\lambda_{\mathrm{min}}^{(\mathrm{HS})}(C_E))^2 |I_i|^2} \cdot (n_{\mathrm{max}}+1)  ||X||_\infty^2\\
    &\le 
    \frac{ N (n_{\mathrm{max}}+1) L^2}{|I_i|^2}\|X\|_{\infty}^2,
    \label{eq:bound_hs_final}
\end{aligned}
\end{equation}
where we have used the fact that $\Pi_{i,k}\le \mathbb{I}/N$ and $||\Pi_{i,k}||_{\mathrm{HS}}^2\le 1/N$.  A possible choice is that $|I_i|$ is independent of $i$, i.e., $|I_i| = L/M$. Then 
\begin{align}
    |\text{Tr}(C_E^{-1}(\Pi_{i,k}/|I_i|) X)|^2 &\le N (n_{\mathrm{max}}+1) M^2 \|X\|_{\infty}^2.
\end{align}
Therefore,
\begin{equation}
\begin{aligned}
\tilde{X}  & \leq \sum_{i=1}^M \sum_{k=1}^N N(n_{\mathrm{max}}+1) M^2 \|X\|_{\infty}^2 \Pi_{i,k} \\
& = N(n_{\mathrm{max}}+1) M^2 \|X\|_{\infty}^2 \sum_{i=1}^M \sum_{k=1}^N \Pi_{i,k} \\
&= N(n_{\mathrm{max}}+1) M^2 \|X\|_{\infty}^2\mathbb{I}
\end{aligned}
\end{equation}
Taking the maximum eigenvalue
\begin{align}
\| X \|_E = \lambda_{\mathrm{max}} \{ \tilde{X} \} \leq N(n_{\mathrm{max}}+1) M^2\|X\|_{\infty}^2
\end{align}

Combining the results from Step 1 and Step 2, we conclude that:
\begin{align}
\text{var}(\hat{X}) \leq \| X \|_E \leq N(n_{\mathrm{max}}+1) M^2 \|X\|_{\infty}^2.
\end{align}
This completes the proof. 
\end{proof}
According to Theorem~\ref{theorem:1}, we only require $N \sim \mathcal{O}(n_{\mathrm{max}})$ and $M \sim \mathcal{O}(n_{\mathrm{max}})$ for the POVM to be informationally complete. In this case,
\begin{align}\label{eq: final_variance}
   \text{var}(\hat{X}) \leq \| X \|_E \leq \mathcal{O}(n_{\mathrm{max}}^4) \|X\|_{\infty}^2. 
\end{align}
It is instructive to compare the sample complexity derived from our direct variance analysis with that implied by the general trace norm bounds for homodyne shadow tomography, such as those presented in~\cite{gandhari2024precision}. The framework in~\cite{gandhari2024precision} guarantees a trace norm error $||\hat{\sigma}^{n_{\mathrm{max}}} - \rho^{n_{\mathrm{max}}}||_1 \le \epsilon_{\mathrm{Tr}}$ (where the state is also projected onto the subspace with photon number cutoff $n_{\mathrm{max}}$) with a sample complexity $T \sim \mathcal{O}(n_{\mathrm{max}}^{13/3} \ln n_{\mathrm{max}} / \epsilon_{\mathrm{Tr}}^2)$. Using Hölder's inequality $|\text{Tr}(X\hat{\sigma}^{n_{\mathrm{max}}}) - \text{Tr}(X\rho^{n_{\mathrm{max}}})| \le ||X||_\infty ||\hat{\sigma}^{n_{\mathrm{max}}} - \rho^{n_{\mathrm{max}}}||_1$, achieving an absolute error $\Delta$ for the expectation value of observable $X$ requires sample complexity $T \sim \mathcal{O}(n_{\mathrm{max}}^{13/3} \ln n_{\mathrm{max}} ||X||_\infty^2 / \Delta^2)$. In contrast, our variance analysis (Eq.~\eqref{eq: final_variance}) suggests a sample complexity $T \sim \mathcal{O}(n_{\mathrm{max}}^4 ||X||_\infty^2 / \Delta^2)$ to achieve the same $\Delta$. This means that, the $\mathcal{O}(n_{\mathrm{max}}^4)$ scaling from our direct variance approach is better than the $\mathcal{O}(n_{\mathrm{max}}^{13/3} \ln n_{\mathrm{max}}) \approx \mathcal{O}(n_{\mathrm{max}}^{4.33} \ln n_{\mathrm{max}})$ scaling implied by the worst-case trace norm bounds, suggesting a potential advantage in sample efficiency when estimating specific observables using our discretized POVM framework, complementing other advanced tomography methods like neural network quantum state tomography or compressed sensing~\cite{torlai2018neural, carrasquilla2019reconstructing, gross2010quantum, ohliger2013efficient}.
\section{Multi-mode scenario}

In this section, we extend our framework from a single-mode continuous-variable system to a multi-mode scenario, which is more representative of complex quantum optical systems and platforms for quantum information processing. We construct the multi-mode classical shadow estimator and analyze its variance.

\subsection{Estimator Construction for Multi-mode Systems}

Consider a system composed of $S$ distinct optical modes. The Hilbert space of this composite system is the tensor product of the individual mode spaces: $\mathcal{H} = \bigotimes_{j=1}^{S} \mathcal{H}_j$. The state of the system is described by a density operator $\rho$ acting on $\mathcal{H}$. To characterize this multi-mode state, we perform  homodyne measurements on each mode independently.

The measurement process for the entire system is described by a set of POVM elements constructed from the tensor product of the single-mode POVMs defined in Section \ref{sec:discretePOVM}. A measurement outcome is given by the tuple $(\vec{i}, \vec{k}) = ((i_1, \dots, i_S), (k_1, \dots, k_S))$, corresponding to observing the quadrature outcome in bin $I_{i_j}$ with phase setting $\theta_{k_j}$ for each mode $j$. The corresponding POVM element is
\begin{equation}
    \Pi_{\vec{i}, \vec{k}} = \bigotimes_{j=1}^{S} \Pi_{i_j, k_j}^{(j)}.
\end{equation}

Experimentally, extending the protocol to multi-mode systems presents significant but manageable challenges: this requires deploying an array of parallel homodyne detectors, each with an independently phase-controlled local oscillator. Crucially, precise synchronization of data acquisition across all modes is necessary to capture a valid snapshot of the multi-mode state, and maintaining phase stability for each channel is key to fulfilling the protocol's statistical assumptions. Despite the increased hardware complexity, such parallel architectures are increasingly feasible in integrated photonic circuits and other advanced quantum platforms.

Following a measurement that yields the outcome $(\vec{i}, \vec{k})$, we construct a "snapshot" of the quantum state serving as an unbiased estimator $\hat{\rho}_{\vec{i}, \vec{k}}$,
\begin{equation}
    \hat{\rho}_{\vec{i}, \vec{k}} = \bigotimes_{j=1}^{S} \left(C_{E,j}^{-1}\left(\frac{\Pi_{i_j, k_j}^{(j)}}{|I_{i_j}|}\right)\right).
\end{equation}
where $C_{E,j}^{-1}$ is the inverse of the linear map for the $j$-th mode and the unbiasedness can be proven similarly as the single mode scenario. The final estimate of the quantum state, $\hat{\rho}$, is obtained by averaging many such snapshots from repeated measurements.

\subsection{Variance Analysis and Sample Complexity}

To estimate the expectation value of an observable $X$ acting on the $S$-mode system, we compute $\text{Tr}(X\hat{\rho})$. The variance of this estimation is bounded by the multi-mode shadow norm $\|X\|_{E, \text{multi}}$.

Let us consider an observable $X$ that can be decomposed into a tensor product of single-mode operators: $X = \bigotimes_{j=1}^{S} X_j$. The variance of its estimator can be shown to be bounded by:
\begin{equation}
    \mathrm{var}(\hat{X}) \le \|X\|_{E, \text{multi}} = \prod_{j=1}^{S} \|X_j\|_{E,j}
\end{equation}
where $\|X_j\|_{E,j}$ is the single-mode shadow norm for operator $X_j$ in mode $j$. It is straightforward that the multi-mode variance scales exponentially with the number of modes, which coincides with the continuous phase settings \cite{gandhari2024precision}.

A key advantage of the shadow estimation framework becomes evident when considering local observables. A $k$-local observable is an operator that acts non-trivially only on $k$ modes and as the identity on the remaining $S-k$ modes. For such an observable, the shadow norm depends only on the $k$ modes, making the sample complexity required to estimate its expectation value independent of the total system size $S$. This scalability is crucial for characterizing large quantum systems efficiently. For an observable $X$ acting non-trivially on a subset of modes $V \subset \{1, \dots, S\}$, the required number of measurements $T$ to achieve a precision $\Delta$ scales as:
\begin{equation}
    T \sim \mathcal{O}\left(\frac{\prod_{j \in V} n^4_{\mathrm{max},j}}{\Delta^2}\right)
\end{equation}
This demonstrates the efficiency of our discretized protocol for estimating local properties of multi-mode continuous-variable quantum states.

\section{Numerical simulation}

\begin{figure}[htbp] 
    \centering 

    \begin{subfigure}{0.46\textwidth} 
        \centering 
        \includegraphics[width=\linewidth]{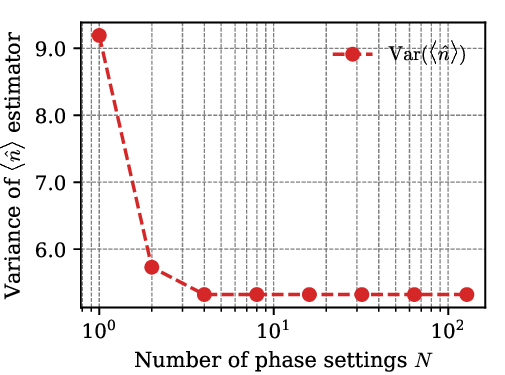}
        \caption{} 
        \label{fig:subfig_a} 
    \end{subfigure}
    \hfill 
    \begin{subfigure}{0.46\textwidth} 
        \centering
        \includegraphics[width=\linewidth]{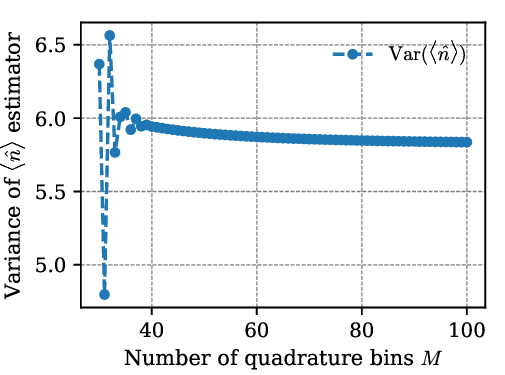}
        \caption{} 
        \label{fig:subfig_b}
    \end{subfigure}

    \begin{subfigure}{0.46\textwidth} 
        \centering
        \includegraphics[width=\linewidth]{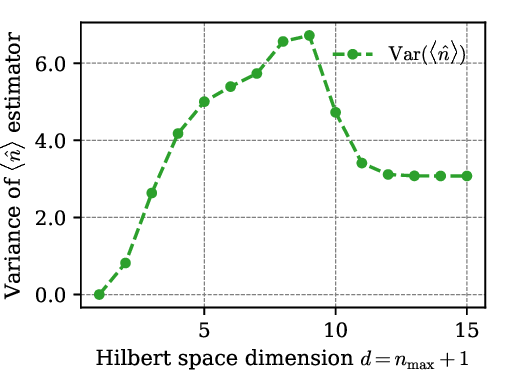}
        \caption{} 
        \label{fig:subfig_c}
    \end{subfigure}
    \hfill 
    \caption{Variance of the estimator of expectation $\langle \hat{n} \rangle = \langle a^\dag a \rangle$ versus system parameters $N$, $M$ and $n_{\mathrm{max}}$.} 
    \label{fig:main_figure} 
\end{figure}
In this section, we investigate numerically how the variance of a single shot estimate varies with different system parameters. We choose a coherent state with mean photon number $|\alpha|^2=1$ and employ our method to construct classical shadows to estimate the expectation value of the mean photon number, i.e., $\hat{n}=a^\dag a$. In Fig.~\ref{fig:subfig_a}, with $M=50$ quadrature bins and $n_{\mathrm{max}}=5$, we vary the number of discrete phases $N$ and compute the variance of the estimate. We observe that the variance saturates quickly and shows almost no dependence on the number of discrete phases when $N \sim 10^1$. For the first few points, the POVM is not informationally complete so that the linear map $\mathcal{C}_E$ is not invertible. One can still define an estimator by replacing the inverse $\mathcal{C}_E^{-1}$ with the Moore–Penrose pseudoinverse $\mathcal{C}_E^{+}$. This yields a well-defined snapshot
\[
\hat{\rho}_{i,k} = \mathcal{C}_E^{+}\!\left(\frac{\Pi_{i,k}}{|I_i|}\right),
\]
which is the minimum-Hilbert–Schmidt-norm solution to the equation $\mathcal{C}_E(\sigma) = \Pi_{i,k}/|I_i|$ in the least-squares sense. While this allows numerical evaluation of observable estimates for any choice of $N$ and $M$, the resulting estimator is generally biased unless $\mathcal{C}_E$ is full rank (i.e., the POVM is informationally complete). We include these initial points even though they do not meet the informational completeness requirement to clearly illustrate the trend of variance as $N$ increases and identify the point where the variance becomes flat. 

The effect of the number of quadrature bins $M$ was also investigated. As shown in Fig.~\ref{fig:subfig_b}, keeping $N$ fixed at $32$ and $n_{\mathrm{max}}=5$, we increased $M$ and observed a gradual decline in the variance. Notably, when $M$ is small (near the informational completeness threshold), the variance exhibits fluctuations before declining gradually. This initial jitter arises from some interconnected factors, all rooted in the discretized nature of the POVM and its dependence on quadrature binning.

 We further illustrate the dependence of variance on the Fock space truncation dimension $n_\mathrm{max}$. As shown in Fig.~\ref{fig:subfig_c}, using our method with $M=100$ fixed quadrature bins and $N=32$ discrete phases, the variance first increases and then drops significantly to a plateau as the truncation dimension grows from $1$ to $15$. We intuitively explain the variance behavior versus different parameters. The decrease in the variance with increasing $N$, $M$ stems from improved measurement precision which collectively lower the uncertainty of measurement outcomes.
 The initial increase in variance with $n_{\mathrm{max}}$ arises because expanding the truncated Fock space dimension raises the information required for state characterization. While the variance decreases as higher-photon-number components of the signal state (here a coherent state) become negligible and the reduced truncation error improves the accuracy of snapshots.

\section{Conclusion}

In this work, we establish a framework for practical shadow estimation in continuous-variable quantum optical systems. We have developed an experimentally feasible shadow protocol using discretized homodyne detection with discrete phase modulation and quadrature bins, constructing an unbiased estimator tailored to realistic experimental constraints. Our central theoretical contributions include deriving both sufficient conditions and necessary conditions for informational completeness, which fundamentally characterize when the discretized measurement can uniquely identify quantum states within the truncated Fock space. Furthermore, we performed comprehensive variance analysis, establishing the favorable $\mathcal{O}(n_{\mathrm{max}}^4)$ scaling of the variance upper bound that improves upon prior $\mathcal{O}(n_{\mathrm{max}}^{13/3})$ results. Numerical simulations validated our theoretical framework, demonstrating systematic reduction with increased discrete phases, quadrature bin resolution, and photon number cutoff.

This work opens several promising research directions for continuous-variable shadow estimation. A natural extension involves developing similar experimental-friendly frameworks for other key CV measurement paradigms, including discretized versions of photon-number resolving detection that account for finite photon resolution and saturation effects, as well as binning strategies for heterodyne detection outcomes that simultaneously measure conjugate quadratures \cite{rohde2015prospects, lita2008counting, yuen1980optical, d1999quantum}. Beyond measurement adaptations, significant potential exists in exploring CV shadow applications mirroring their discrete-variable counterparts, particularly in quantum error mitigation techniques like virtual distillation and probabilistic error cancellation adapted to infinite-dimensional spaces, and in efficient entanglement detection protocols for continuous-variable systems \cite{li2017efficient, temme2017error, ralph2011continuous, duan2000inseparability, simon2000peres, giovannetti2003characterizing}. Further opportunities include combining federated learning framework \cite{song2024quantum}, integrating machine learning or online learning \cite{jiang2025shadow} techniques for enhanced state reconstruction, and implementing real-time shadow processing in photonic and superconducting quantum platforms to characterize non-Gaussian states and operations \cite{wang2020integrated, blais2021circuit}.

\section*{Acknowledgements}
This work was supported in part by the National Natural Science Foundation of China Grants No. 92465202, 62325210, 12447107, 12204489, 62301531, 12501450, and Beijing Natural Science Foundation No. 4252013, 4252012.

\bibliography{bibcvshadow}

\end{document}